\documentclass{article}

\usepackage[utf8]{inputenc}
\usepackage{amsmath,amsthm,amssymb}
\usepackage[margin=3cm]{geometry}
\usepackage{graphicx}
\usepackage{tikz}
\usepackage{tikz-cd}
\usepackage{hyperref}
\usepackage{cleveref}
\usepackage{adjustbox}
\usepackage{array}
\usepackage{authblk}
\usepackage{titlesec}
\usepackage{caption} 
\usepackage{algorithm}
\usepackage[noend]{algpseudocode}

\def\algbackskip{\hskip-\ALG@thistlm}

\newtheorem{theorem}{Theorem}[section]
\newtheorem{corollary}{Corollary}[theorem]
\newtheorem{lemma}[theorem]{Lemma}
\newtheorem{proposition}[theorem]{Proposition}
\theoremstyle{definition}
\newtheorem{definition}[theorem]{Definition}

\titlespacing*{\section}
{0pt}{5.5ex plus 1ex minus .2ex}{4.3ex plus .2ex}
\titlespacing*{\subsection}
{0pt}{5.5ex plus 1ex minus .2ex}{4.3ex plus .2ex}

\newcommand{\fO}{f_{\restriction_\Omega}}
\newcommand{\fT}{f_{\restriction_\Theta}}
\newcommand{\B}{\mathbb{B}}

\title{Control Strategy Identification via Trap Spaces \\ in Boolean Networks}

\author[1,2]{Laura Cifuentes Fontanals} 
\author[1]{Elisa Tonello} 
\author[1]{Heike Siebert}
\affil[1]{Freie Universität Berlin, Germany}
\affil[2]{Max Planck Institute for Molecular Genetics, Berlin, Germany}

\date{}

\begin{document}

\maketitle

\begin{abstract}
The control of biological systems presents interesting applications such as cell reprogramming or drug target identification. A common type of control strategy consists in a set of interventions that, by fixing the values of some variables, force the system to evolve to a desired state. This work presents a new approach for finding control strategies in biological systems modeled by Boolean networks. In this context, we explore the properties of trap spaces, subspaces of the state space which the dynamics cannot leave. Trap spaces for biological networks can often be efficiently computed, and provide useful approximations of attraction basins. Our approach provides control strategies for a target phenotype that are based on interventions that allow the control to be eventually released. Moreover, our method can incorporate information about the attractors to find new control strategies that would escape usual percolation-based methods. We show the applicability of our approach to two cell fate decision models.
\end{abstract}

\section{Introduction}

The control of biological systems presents interesting applications such as cell fate reprogramming, drug target identification for disease treatments or stem cells programming \cite{DrugDiscovery,Reprogramming}. Controlling a cell fate decision network could for instance allow, in the case of cancer cells, to lead the system to an apoptotic state and, therefore, evolve towards the elimination of pathological cells \cite{ApoptosisCancer}. Finding adequate candidates for control is a complex problem, in particular since the experimental testing of all the possibilities is not feasible. Mathematical modeling can help address this problem by enabling \textit{in silico} identification of possible effective candidates.

Modeling of biological processes is often challenged by the lack of information about kinetic parameters or specific reaction mechanisms. The Boolean formalism aims at capturing the qualitative behavior of systems via a coarse representation of the relationship between the species of interest. Mechanisms underlying activation and inhibition processes are summarized in logical functions, allowing for two activity levels for each variable. The two values can represent for example if a gene is expressed or not, or if the concentration of a protein is above or below a certain threshold. Boolean modeling has in many instances been shown to capture the fundamental behaviors and dynamics of biological systems and has been widely used to make predictions or design strategies for therapeutic interventions \cite{CellFate_network,Sinergies,MAPK_network}.

Control of biological systems is a broad field that encompasses a variety of approaches and goals. Attractor control aims at leading the system to a desired attractor, starting from a particular initial state (source-target control) \cite{ControlBasins} or from all possible initial states (full-network control) \cite{ControlMotifs}. However, it is often useful to induce a desired phenotype rather than a specific attractor. Phenotypes are usually defined in terms of some biomarkers i.e., observable and measurable components that represent the main characteristics of biological processes. The approach that focuses control on a set of relevant variables is also known as target control \cite{InterventionSets,TargetControl}. In this work, we are interested in full-network control for a target phenotype.

There are different approaches for system interventions, that is, the way the control is applied to biological systems. In the context of Boolean modeling, we consider as interventions the perturbations or modifications that fix the value of some components (node control) \cite{ControlBasins,ControlMotifs}. In the example of a gene regulatory network, fixing a variable to a certain value can be understood as the knockout or permanent activation of a gene. Among other approaches to Boolean network control is edge control, which targets the interactions between variables \cite{ControlBCN,ControlAlgebra}. For a gene regulatory network, edge control can be interpreted for instance as the modification of a protein to alter its interaction with a certain gene.

Control of dynamical systems has been a popular research field in systems biology in the last years, also in the Boolean setting. Many approaches focus on the structure and topology of the network, for example by looking at feedback loops \cite{FeedbackVertexSet} or stable motifs \cite{ControlMotifs}, and several studies discuss the complexity and characteristics of such problems \cite{KernelControl,Controllability_Networks}. Other approaches include techniques based on topological information to reduce the size of the search space \cite{InterventionSets} or computational algebra methods \cite{ControlAlgebra}. Recent works have explored attractor control through the characterization of basins of attraction, that is, sets of states from which only a certain attractor can be reached \cite{ControlBasins}. However, the identification of basins of attraction might require the exploration of the complete state space. Attractor reachability can be investigated using trap spaces, which are subspaces that trajectories cannot leave. By definition, every trap space contains at least one attractor and, therefore, in some cases minimal trap spaces can be good approximations for the attractors \cite{AttractorApprox_Klarner}. The identification of trap spaces in biological systems can often be performed efficiently by exploiting properties of the prime implicants \cite{Implicants_Klarner}. 

Our approach aims to identify strategies for phenotype control by exploiting properties of trap spaces. We introduce the concept of space of attraction, a subspace that approximates the basin of attraction, to find control strategies without the need of computing the whole basin. We extend this idea to define spaces of attraction for trap spaces and relate them to control strategies, which are defined as sets of constraints that fix the value of some variables and induce a certain target phenotype. We exploit properties of trap spaces and computation techniques for target control to define a new method to compute control strategies that do not require a permanent intervention and allow the control to be eventually released. Our approach can incorporate information about the attractors to obtain new control strategies that might escape percolation-based target control techniques. The method presented here is widely applicable to Boolean models of biological systems and can provide, under certain conditions, control strategies that are independent of the type of update used in the model. 

We start by giving a general overview about Boolean modeling (\Cref{Background}). Then we introduce the concepts of control strategy and space of attraction in this setting (\Cref{SoA and CA}), providing the theoretical bases for the computation of some types of control strategies. In \Cref{Computation}, we present a method to compute control strategies based on the theoretical principles explained in \Cref{SoA and CA} and implemented using the prime implicants of the function. Lastly, in \Cref{Application} we show the applicability of our method to two cell fate decision networks \cite{MAPK_network,TLGL_network}.

\section{Background: Boolean networks and dynamics} \label{Background}

A \textit{Boolean network} on $n$ variables is defined as a function $f\colon \B^n \rightarrow \B^n$, where $\B = \{0,1\}$. $V = \{1,...,n\}$ is the set of variables of $f$, $\B^n$ is the \textit{state space} of the Boolean network and every $x \in \B^n$ is a \textit{state} of the state space. For any $x \in \B^n$ and $I \subseteq V$, $\bar{x}^I$ is defined as $\bar{x}^I_i = x_i$ for $i \in V \backslash I$ and $\bar{x}^I_i = 1 - x_i$ for $i \in I$. If $I = \{i\}$, $\bar{x}^I$ is written as $\bar{x}^i$.

A \emph{dynamics} on $\B^n$ or \emph{state transition graph} is a directed graph with vertex set $\B^n$. There are several ways of associating a dynamics to a Boolean network $f$. In the \emph{general asynchronous dynamics} or \emph{general asynchronous state transition graph} $GD(f)$ there exists an edge from a vertex $x$ to a vertex $y$ if and only if there exists $\emptyset \neq I \subseteq V$ such that $\bar{x}^I = y$ and $f_i(x) = y_i$ for every $i \in I$. Note that the general asynchronous dynamics considers transitions which update subsets of components simultaneously in a non-deterministic way. By choosing different types of updates, other state transition graphs can be defined. The \emph{asynchronous dynamics} $AD(f)$ is defined by considering the transitions updating only one component at a time and the \emph{synchronous dynamics} $SD(f)$ considers only the transitions where all the components that can be updated are updated at once. Note that $AD(f)$ and $SD(f)$ are subgraphs of $GD(f)$. To simplify the notation, $D(f)$ will denote any of these dynamics associated to $f$. The choice of asynchronous and general asynchronous updates is motivated by the attempt to capture different, and sometimes unknown, time scales that might coexist in the modeled system. An example of asynchronous dynamics of a Boolean network is shown in \Cref{ex_space_basin_attraction}.

A \emph{trap set} $T \subseteq \B^n$ is a set such that for all $x \in T$, if $y$ is a successor of $x$ in the dynamics, then $y \in T$. A minimal trap set under inclusion is an \emph{attractor}. An attractor can be a \emph{stable state} (or \emph{fixed point}), when it consists only of one state, or a \emph{cyclic} (or \emph{complex) attractor} when it is larger. In biological systems, stable states can be identified with different cell fates or cell types, and cyclic attractors with cell cycles or specific cell processes. Given a Boolean function $f$ and an attractor $A$, the \textit{weak basin of attraction} of $A$ is defined as the set of states $x$ such that there exists a path from $x$ to an element of $A$ in $D(f)$. The \textit{strong basin of attraction} of $A$ is the set of states in the weak basin of $A$ that do not belong to the weak basin of attraction of any other attractor different from $A$. \Cref{ex_space_basin_attraction} shows the weak and strong basins for an attractor in an asynchronous state transition graph.

The control interventions considered in this work consist in fixing the values of some components. Formally, given a state $c \in \B^n$ and a subset of variables $I \subseteq V$, we define the \emph{subspace} induced by $c$ and $I$ as the set $\Sigma(I,c) = \{x \in \B^n |\text{ }\forall i \in I, x_i = c_i \}$. The variables in $I$ are called \emph{fixed variables}, while the other variables are called \emph{free}. We denote subspaces as states, using the symbol $*$ for the free variables. For example, the subspace $\{ x \in \B^4 | x_1 = 1 \text{ and } x_3 = 0\}$ is denoted as $1*0*$. 

The identification of control variables requires examining the effect that fixing certain variables has on the dynamics. Given a Boolean function $f$ and a subspace $\Theta = \Sigma(I,c)$, the restriction of the function $f$ to the subspace $\Theta$ is defined as:
$$
\fT\colon \Theta \rightarrow \Theta,
\text{ where for all } i \in V \text{, } (\fT)_i(x) = \left\{
\begin{array}{ll}
f_i(x), & i \notin I, \\
c_i, &  i \in I. \\
\end{array}
\right.
$$
Note that $\fT\colon \Theta \rightarrow \Theta$ can be identified with a Boolean network $g\colon \B^m \rightarrow \B^m$, where $m = n - |I|$. Via this identification, we extend all the definitions that apply to a Boolean network to such restrictions. For example, the state transition graph corresponding to $\fT\colon \Theta \rightarrow \Theta$ is defined as usual, only with vertex set $\Theta$ instead of $\B^n$ (see \Cref{ex_cs_not_found_by_perc}). Moreover, if $T$ is a trap set in $D(f)$, then $T \cap \Theta$ is a trap set in $D(\fT)$.

A subspace that is also a trap set is called a \emph{trap space}.
While trap sets and attractors might vary when considering different types of dynamics, trap spaces are independent of the type of update. The Boolean function represented in \Cref{ex_space_basin_attraction} has four trap spaces: $000$, $111$, $0*0$, $***$.

In this work we aim at using trap spaces to find control strategies for phenotypes. Phenotypes are usually defined in terms of the state of some measurable components called biomarkers, which are observable components that can be used as indicators of different cell types or cell fates or to distinguish between healthy and pathological conditions. Although the notion of phenotype is usually related to stability, we extend this concept to consider any possible state in order to allow non-attractive states satisfying the phenotype characteristics to become attractors in the controlled system. Thus, in this work, we define a \emph{phenotype} as a subspace.

\begin{figure}
\hspace{0.3cm}\begin{minipage}{0.25\linewidth}
\begin{center}
\tikz[overlay]{
\filldraw[fill = red!10, draw=red, thick] (-1.87,-2.6) rectangle (-1.22,-2.25);
\filldraw[fill = blue!10, draw=blue] (1.17,-0.25) rectangle (1.78,-0.65);}
\begin{tikzcd}[column sep=5, row sep=5]
& & & & & \\
& {\color{orange} 110} \arrow[ld, shift right = 1] \arrow[rr] & & 111 \\
{\color{orange} 100} \arrow[ur, shift right = 1] \arrow[dd] & & {\color{red} 101} \arrow[dd, shift right=1] &  \\
& {\color{red} 010} \arrow [ld] & & {\color{red} 011} \arrow[ld] \\
{\color{red} 000} & & {\color{red} 001} \arrow[ll] \arrow[uu, shift right=1] & \\
\end{tikzcd}
\end{center}
\end{minipage} \hspace{0.55cm}
\begin{minipage}{0.66\linewidth}
{\small
Basins of attraction of $A_1$:
\begin{itemize}
\item $Strong(A_1) = \{000, 001, 010, 011, 101 \}$
\item $Weak(A_1) = \{000, 001, 010, 011, 101, 100, 110 \}$ 
\end{itemize}
Spaces of attraction of $A_1$:
\begin{itemize}
\item $\Omega_1 = 0**$, $\Omega_2 = 00*$, $\Omega_3 = 01*$, $\Omega_4 = 0*0$, $\Omega_5 = 0*1$, $\Omega_6 = *01$, $\Omega_7 = 000$, $\Omega_8 = 001, \Omega_9 = 010$, $\Omega_{10} = 011$, $\Omega_{11} = 101$, with $\Omega_i \subsetneq Strong(A_1)$ for all $1 \leq i  \leq 11$.
\end{itemize}
}
\end{minipage}
\caption{Asynchronous dynamics of the Boolean function $f(x) = (\bar{x}_1 \bar{x}_2 x_3 \lor x_1 x_2$, $x_1 \bar{x}_2 \bar{x}_3 \lor x_1 x_2 x_3$, $x_1 x_2 \lor x_1 x_3 \lor x_2 x_3)$, with attractors $A_1 = 000$ and $A_2 = 111$ and trap spaces $000$, $111$, $0*0$, $***$. All the spaces of attraction of $A_1$ are included in its strong basin (in red) while the basin itself is not a space of attraction.}
\label{ex_space_basin_attraction}
\end{figure}

\section{Spaces of attraction and control strategies} \label{SoA and CA}

The strong basin of attraction of an attractor $A$ can be naturally related to control since, by definition, it contains all the states that have paths to $A$ but not to any other attractor.
In contrast to methods requiring basin exploration, we use subspace approximation of the basins combined with trap spaces computation. To do so, we extend the notion of basin of attraction to trap sets. We then exploit useful properties of trap spaces, e.g. independence of the update, efficient identification and potential approximation of attractors, to develop a new approach for the identification of control strategies.

\subsection{Control strategies}\label{sub:controlstrategies}

We now formalise the notion of control strategy for a phenotype. A control strategy is a subspace defined by a set of interventions that fix the value of some variables and thus force all attractors to be contained in the subspace defining the phenotype.

\begin{definition}
\label{def:cs}
Given a Boolean function $f$ and a phenotype $P \subseteq \B^n$, a \emph{control strategy (CS) for the phenotype $P$} in $D(f)$ is a subspace $\Theta \subseteq \B^n$ such that, for any attractor $A$ of $D(\fT)$, $A \subseteq P$.
\end{definition}

If the desired phenotype is a stable state in the original dynamics $(P = \{y\}$, $y \in \B^n)$, a control strategy for $P$ is a subspace $\Theta$ such that $y$ is the only attractor of $\fT$. \Cref{ex_cs_not_found_by_perc} shows an example of a control strategy for a stable state. The size of the subspace defining a control strategy represents the number of interventions in the system. Therefore, the most interesting control strategies are the subspaces that are maximal with respect to inclusion.

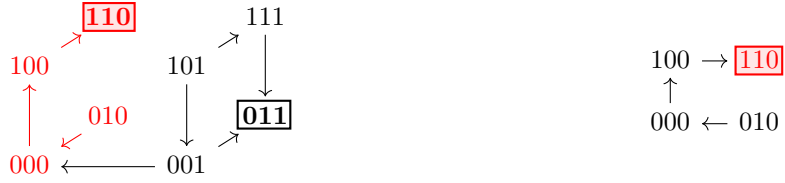
\begin{figure}
\begin{minipage}{0.6\linewidth}
\begin{center}
\tikz[overlay]{
\filldraw[fill = red!10, draw=red, thick] (1.25,0.88) rectangle (1.93,1.25);
\filldraw[fill = white, draw=black, thick] (3.3,-0.42) rectangle (4,-0.05);}
\begin{tikzcd}[column sep=5, row sep=5]
 & {\color{red} \textbf{110}} & & 111 \arrow[dd] \\
{\color{red}100} \arrow[ur, red] & & 101 \arrow[dd] \arrow[ru] & \\
 & {\color{red}010} \arrow[ld, red] & & \textbf{011} \\
{\color{red}000} \arrow[uu, red] & & 001 \arrow[ll] \arrow[ru] \\
\end{tikzcd}
\end{center}
\end{minipage}
\begin{minipage}{0.35\linewidth}
\begin{center}
\tikz[overlay]{
\filldraw[fill = red!10, draw=red, thick] (1.4,0.67) rectangle (2.02,0.3);}
\begin{tikzcd}[column sep=5, row sep=5]
100 \arrow[rr]  & & {\color{red}110} \\
& & \\
000  \arrow[uu] & & 010 \arrow[ll] \\
\end{tikzcd}
\end{center}
\end{minipage}
\caption{Asynchronous dynamics of the Boolean function $f(x) = (x_1 \bar{x}_3 \lor \bar{x}_2 \bar{x}_3$, $x_1 \lor x_3 $, $x_1 x_3 \lor x_2 x_3)$ (left) and $\fO(x) = (x_1 \lor \bar{x}_2$, $x_1$, $0)$ with $\Omega = **0$ (right). $\Omega$ is a control strategy for the phenotype $P = \{110\}$ in $AD(f)$. $\Omega$ does not percolate to $P$.}
\label{ex_cs_not_found_by_perc}
\end{figure}

A common approach in the context of control is the use of value percolation \cite{InterventionSets,TargetControl}. Different combinations of variables to be fixed are considered, and their values propagated iteratively until an invariant subspace is reached. A combination of variables and values is an intervention strategy if the subspace obtained at the end of the iterative percolation process is contained in the target phenotype. Strategies obtained with this approach satisfy the conditions of \Cref{def:cs}. However, the class of control strategies identified by the definition is larger, as we will discuss in the following.

\subsection{Spaces of attraction}

Trap sets are sets of states that the dynamics cannot leave. Each trap set contains, as a consequence, at least one attractor. The concept of basin of attraction defined for an attractor can be naturally extended to trap sets. As mentioned before, we wish to approximate basins of attraction by subspaces. Combining these two ideas, we introduce the concept of space of attraction of a trap set $T$ as a subspace $\Omega$ such that from any state in $\Omega$ there exists a path to $T$ and no trap set disjoint from $T$ is reachable from $\Omega$.

\begin{definition}
\label{def:space_attr}
Let $f$ be a Boolean function and $T$ be a trap set of $f$. A \emph{space of attraction} of the trap set $T$ in $D(f)$ is a subspace $\Omega$ such that for all $x \in \Omega$ and for any trap set $S$, if there exists a path in $D(f)$ from $x$ to an element of $S$, then $S \cap T \neq \emptyset$.
\end{definition}

\Cref{def:space_attr} implies the existence of a path from the space of attraction $\Omega$ to the trap set $T$. A trap set can have many spaces of attraction. In fact, any subspace contained in a space of attraction is also a space of attraction. Moreover, if there is only a unique trap set $T_m$ minimal with respect to inclusion contained in a trap set $T$, any space of attraction of $T$ is also a space of attraction of $T_m$. Both trap spaces and spaces of attraction are subspaces that characterize the long term behavior of the system. However, in contrast to trap spaces, spaces of attraction can depend on the update.

If a trap set is an attractor, its spaces of attraction can be related to its basins of attraction. The spaces of attraction of an attractor $A$ are clearly contained in the strong basin of $A$ since, by \Cref{def:space_attr}, none of the other attractors can be reached from any state inside the space of attraction. However, the strong basin of attraction of $A$ might not be a space of attraction (see \Cref{ex_space_basin_attraction}).

Spaces of attraction, as well as basins, might include paths crossing non-attractive cycles in the state transition graph. As a consequence, some paths starting in the space of attraction (or basin) might not reach the trap space (or attractor), staying indefinitely in non-attractive cycles.  While in very specific circumstances such behavior might be relevant, generally it constitutes an artifact arising from the non-deterministic update. Here, we extend the standard view on basins of attraction to spaces of attraction, assuming the trajectories of interest will eventually leave non-attractive strongly connected components in the state transition graph.

The condition that a subspace needs to satisfy to be a space of attraction of a trap set $T$ gets simplified when the subspace considered is the entire state space. In this case, it is only required that from every state in the state space trap set $T$ can be reached (see \Cref{soa_restr}), since it immediately follows that there cannot be a trap set disjoint from $T$. 

\begin{lemma}
\label{soa_restr}
Let $f$ be a Boolean function and $T$ a trap set of $f$. Then $\B^n$ is a space of attraction of the trap set $T$ in $D(f)$ if and only if for all $x \in \B^n$ there exists a path in $D(f)$ from $x$ to some $y \in T$.
\end{lemma}

The application of \Cref{soa_restr} to the restriction on a subspace immediately yields the following corollary.

\begin{corollary}
\label{space_attr_restr}
Let $f$ be a Boolean function, $T$ a trap set of $f$ and $\Omega$ a subspace such that $T \subseteq \Omega$. Then $\Omega$ is a space of attraction of $T$ in $D(\fO)$ if and only if for all $x \in \Omega$ there exists a path in $D(\fO)$ from $x$ to some $y \in T$.
\end{corollary}

In other words, a space of attraction of a trap set for the Boolean function restricted to that subspace can be understood as the restrictions that we can impose on the function $f$ to lead the dynamics to a certain trap set. If $T$ is a trap space, there is always a trivial space of attraction for the restricted function which is $T$ itself.

Note that a subspace $\Omega$ that is a space of attraction of $T$ for the Boolean function $f$ is not necessarily a space of attraction for the restricted function $f_{\restriction_\Omega}$ (see \Cref{Example}).

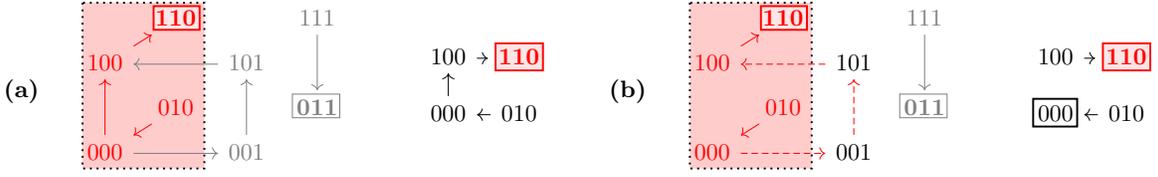
\begin{figure} 
\textbf{\small(a)} \hspace{-0.32cm} \begin{minipage}{0.3\linewidth}
\adjustbox{scale=0.9,center}{
\tikz[overlay]{
\filldraw[fill = red!20,draw=black, thick, dotted] (0.2,1.3) rectangle (2,-1.15);
\filldraw[fill = red!10, draw=red, thick] (1.25,1.25) rectangle (1.93,0.9);
\filldraw[fill = white, draw=gray] (3.3,-0.05) rectangle (4,-0.4);}
\begin{tikzcd}[column sep=5, row sep=5]
 & {\color{red} \textbf{110}} & & {\color{gray}111} \arrow[dd, gray] \\
{\color{red}100} \arrow[ur, red] & & {\color{gray} 101} \arrow[ll, gray] & \\
 & {\color{red}010} \arrow[ld, red] & & {\color{gray}\textbf{011}} \\
{\color{red}000} \arrow[uu, red] \arrow[rr, gray] & & {\color{gray} 001} \arrow[uu, gray] \\
\end{tikzcd}
}
\end{minipage}
\begin{minipage}{0.15\linewidth}
\adjustbox{scale=0.9,center}{
\tikz[overlay]{
\filldraw[fill = red!10, draw=red, thick] (1.23,0.7) rectangle (1.93,0.3);}
\begin{tikzcd}[column sep=5, row sep =5]
100 \arrow[r] & {\color{red}\textbf{110}} \\
& \\
000 \arrow[uu] & 010 \arrow[l] \\
\end{tikzcd}
}
\end{minipage}
\hspace{0.3cm} \textbf{\small (b)} \hspace{-0.32cm} \begin{minipage}{0.3\linewidth}
\adjustbox{scale=0.9,center}{
\tikz[overlay]{
\filldraw[fill = red!20,draw=black, thick, dotted] (0.2,1.3) rectangle (2,-1.15);
\filldraw[fill = red!10, draw=red, thick] (1.25,1.25) rectangle (1.93,0.9);
\filldraw[fill = white, draw=gray] (3.3,-0.05) rectangle (4,-0.4);}
\begin{tikzcd}[column sep=5, row sep=5]
 & {\color{red} \textbf{110}} & & {\color{gray}111} \arrow[dd, gray] \\
{\color{red}100} \arrow[ur, red] & & 101 \arrow[ll, red, dashed] & \\
 & {\color{red}010} \arrow[ld, red] & & {\color{gray}\textbf{011}} \\
{\color{red}000} \arrow[rr, red, dashed] & & 001 \arrow[uu, red, dashed] \\
\end{tikzcd}
}
\end{minipage}
\begin{minipage}{0.15\linewidth}
\adjustbox{scale=0.9,center}{
\tikz[overlay]{
\filldraw[fill = white,draw=black, thick] (0.85,-0.53) rectangle (0.2,-0.12);
\filldraw[fill = red!10, draw=red, thick] (1.23,0.7) rectangle (1.93,0.3);}
\begin{tikzcd}[column sep=5, row sep=5]
100 \arrow[r] & {\color{red}\textbf{110}} \\
& \\
000 & 010 \arrow[l] \\
\end{tikzcd}
}
\end{minipage}
\caption{\textbf{(a)} $\Omega = **0$ is a space of attraction for $AD(f)$ and $AD(\fO)$, with $f(x) = (\bar{x}_2 \lor x_1 \bar{x}_3, x_1 \bar{x}_3 \lor x_2 x_3, \bar{x}_1 \bar{x}_2 \lor x_2 x_3)$ and $\fO(x) = (\bar{x}_2 \lor x_1, x_1, 0)$. \textbf{(b)} $\Omega = **0$ is a space of attraction for $AD(g)$ but not for $AD(g_{\restriction_\Omega})$, with $g(x) = (x_1 \bar{x}_3 \lor \bar{x}_2 x_3 \lor x_1\bar{x}_2, x_1 \bar{x}_3 \lor x_2 x_3, \bar{x}_1 \bar{x}_2 \lor x_2 x_3)$ and $g_{\restriction_\Omega} (x) = (x_1, x_1, 0)$.}
\label{Example}
\end{figure}

Given a trap space $T$ that only contains attractors belonging to a certain phenotype $P$, any space of attraction that leads the system to $T$ would also lead it to an attractor belonging to $P$. In other words, any space of attraction for a trap space $T$ is also a control strategy for a phenotype $P$ if $T$ only contains attractors belonging to $P$. The following proposition formalizes this idea.

\begin{proposition}
\label{prop_CS}
Let $P \subseteq \B^n$ be a subspace and $f$ a Boolean function. Let $T$ be a trap space such that if $A \subseteq T$ is an attractor of $D(f)$, then $A \subseteq P$. Let $\Omega$ be a space of attraction of $T$ in $D(\fO)$ such that $T \subseteq \Omega$. Then $\Omega$ defines a control strategy in $D(f)$ for $P$.
\end{proposition}

\begin{proof}
Let $A$ be an attractor for $D(\fO)$. Then $A \subseteq \Omega$. Since $\Omega$ is a space of attraction of $T$ in $D(\fO)$ and $A$ is a trap set in $D(\fO)$, $T \cap A \neq \emptyset$. As $T$ and $A$ are trap sets, $T \cap A$ is also a trap set in $D(\fO)$. Since $A$ is minimal,  $A = T \cap A \subseteq T$. Then, since $T$ is a trap space and for all $x \in T, \fO (x) = f(x)$, $A$ is also an attractor of $D(f)$ and, therefore, $A \subseteq P$. 
\end{proof}

Since a trap space is always a space of attraction of itself, given a subset $P \subseteq \B^n$, any trap space $T$ containing only attractors in $P$ is a control strategy for $P$.

The type of control strategies identified by \Cref{prop_CS} allow the interventions to be released after a certain number of steps. That is because these control strategies induce the target phenotype by leading the system to a trap space. Once the trap space is reached, since the dynamics cannot leave it, the control can be released and the system will remain in the trap space, eventually evolving to the phenotype of interest.

\subsection{Identification of spaces of attraction}

As explained in the previous section, control strategies for a phenotype $P$ can be found by identifying spaces of attraction of trap spaces containing only attractors in $P$. In this section, we explore ways of finding spaces of attraction for trap spaces.

Given a trap space $T$, we look for a subspace $\Omega$ such that from all states in $\Omega$ there is a path to $T$ in $D(\fO)$. To do so, we use the idea of value percolation, which is a common approach in the context of control. As explained in \Cref{sub:controlstrategies}, it is based on the fact that the constraints given by the fixed variables of a subspace might induce further variables to get fixed. Thus, in our setting, a subspace $\Omega = \Sigma(W,c)$ that percolates to the trap space $T = \Sigma(U,c)$ is a space of attraction of $T$ in $\fO$. The following lemma formalizes this idea.

\begin{lemma}
\label{lemma_path}
Let $f\colon \B^n \rightarrow \B^n$ be a Boolean function, $c \in \B^n$ and $S = \Sigma(U,c)$, $\Omega = \Sigma(W,c)$ subspaces of $\B^n$ such that $S \subseteq \Omega$ and $W \subseteq U \subseteq V$. If for all $s \in U \backslash W$, $f_s(x) = c_s$ for all $x \in \Omega$, then for all $x \in \Omega$ there exists a path in $D(\fO)$ from $x$ to some $y \in S$.
\end{lemma}

\begin{proof}
Since the proof depends on the update, we treat each case separately.

$D = AD$: For each $x \in \Omega$ and for each $s \in U \backslash W$ such that $x_s \neq c_s$, $f_s(x) = c_s$. Therefore, $x$ admits a successor $y$ in $AD(\fO)$ with $y_s = c_s$. This implies the existence of a path in $AD(\fO)$ from any state in $\Omega$ to $S$.

$D = SD$: For each $x \in \Omega$ and for each $s \in U$, $f_s(x) = c_s$. Therefore, $x$ admits a successor $y \in S \subseteq \Omega$ in $SD(\fO)$. 

$D = GD$: Since all the paths in $AD(f)$ and $SD(f)$ are also paths in $GD(f)$, the conclusion follows from the previous cases.
\end{proof}

\Cref{lemma_path} can be extended with \Cref{space_attr_restr} to provide conditions that allow the identification of spaces of attraction.

\begin{lemma}
\label{lemma_perc_1}
Let $f\colon \B^n \rightarrow \B^n$ be a Boolean function and $T = \Sigma(U, c)$ a trap space of $f$ with $U \subseteq V$ and $c \in \B^n$. Let $\Omega = \Sigma(W, c)$ be a subspace of $\B^n$ such that $T \subseteq \Omega$ and $W \subseteq U \subseteq V$. If $f_s(x) = c_s$ for all $x \in \Omega$ and $s \in U \backslash W$, then $\Omega$ is a space of attraction of $T$ for $D(\fO)$.
\end{lemma}

To improve the spaces of attraction obtained with \Cref{prop_CS}, we can extend \Cref{lemma_perc_1} applying the idea used in \Cref{lemma_path} several times, building a path of percolated subspaces ending in the trap space $T$.

\begin{proposition}
\label{perc_soa}
Let $f\colon \B^n \rightarrow \B^n$ be a Boolean function and let $c \in \B^n$. Let $T = \Sigma(U,c)$  be a trap space and $\Omega = \Sigma(W, c)$ be a subspace containing $T$ with $W \subseteq U \subseteq V$. Let $I_0 = W$ and $I_{k+1} = \{ s \in U | s \in I_k \text{ or } f_s(x) = c_s \text{ for all } x \in S_k \}$, where $S_k = \Sigma(I_k, c)$. If there exists a $k_T$ such that $I_{k_T} = U$, then $\Omega$ is a space of attraction of $T$ for $D(\fO)$.
\end{proposition}

\Cref{perc_soa} gives sufficient conditions for a subspace to be a space of attraction of a trap space in the restriction and, together with \Cref{prop_CS}, provides a way to identify control strategies for a given phenotype. However, not all spaces of attraction fall under the conditions given by \Cref{perc_soa}. The example in \Cref{Example} (a) shows a space of attraction $\Omega = **0$ for a trap space $T = 110$, which is also a control strategy for $P = \{110\}$, where $\Omega$ does not percolate to $T$.

Sometimes the attractors of a system of interest are known. In other cases they are not known but can be approximated by minimal trap spaces \cite{AttractorApprox_Klarner}, that is, each minimal trap space contains only one attractor and every attractor is included in a minimal trap space. This information is not usually exploited by target control methods, which often rely solely on percolation-like techniques. The approach described in this work can use this knowledge to find additional control strategies. If the attractors are known or they can be approximated by minimal trap spaces, we can easily find trap spaces satisfying the conditions of \Cref{prop_CS} by simply checking whether these attractors or minimal trap spaces are included in a trap space. Therefore, larger trap spaces containing only attractors of the target phenotype can be identified. By \Cref{prop_CS}, spaces of attraction for these trap spaces are also control strategies for the phenotype. These control strategies do not necessarily percolate to the phenotype and, therefore, might not be identified by usual percolation techniques. \Cref{ex_cs_not_found_by_perc} shows an example of such a control strategy, where $\Omega = T = **0$ is a space of attraction for the trap space $T$, which contains only the attractor $A = 110$, and so, is a control strategy for the phenotype $P = A$. Note that $\Omega$ does not percolate to $A$.

The attractors of a Boolean network might vary in different dynamics. Therefore, the trap spaces satisfying \Cref{prop_CS} and the control strategies characterized by them might also be dependent on the dynamics. Conversely, the spaces of attraction obtained by \Cref{perc_soa} are independent of the update. Thus, if the trap spaces considered satisfy the conditions of \Cref{prop_CS} in all the dynamics, the control strategies identified are also independent of the update.

\section{Computation of control strategies} \label{Computation}

We propose a method to find control strategies for a given phenotype, using the ideas explained in the previous section. The main steps of the method are represented in \Cref{fig:method} and the detailed procedure is shown in \Cref{alg:ca}.

In order to implement the computation of the control strategies, we use the prime implicants of the function. Given a Boolean function $f\colon\B^n \rightarrow \B^n$, a \emph{$c$-implicant} of $f_i$, with $c \in \B$ and $i \in V$, is a subspace $Q$ such that $f_i(x) = c$ for all $x \in Q$. A \emph{prime implicant} is an implicant that is maximal under inclusion. Given $T = \Sigma(U,c)$, finding a subspace satisfying the hypothesis of \Cref{lemma_perc_1} is equivalent to finding a subspace that is a $c_i$-implicant of $f_i$ for all $i \in U$. Moreover, prime implicants can also be used to compute the trap spaces \cite{Implicants_Klarner}. The computation of the prime implicants of a Boolean function is in general a hard problem. However, networks modeling biological systems are usually relatively sparse, since the number of components regulating a variable is relatively small compared to the size of the network. Therefore, they are rather tractable in terms of prime implicants computation. Several tools are available for the computation of prime implicants and trap spaces of Boolean functions. We use \textit{PyBoolNet} \cite{PyBoolNet}, a Python package that allows generation and analysis of Boolean networks and provides an efficient computation of prime implicants and trap spaces for quite large networks. \\

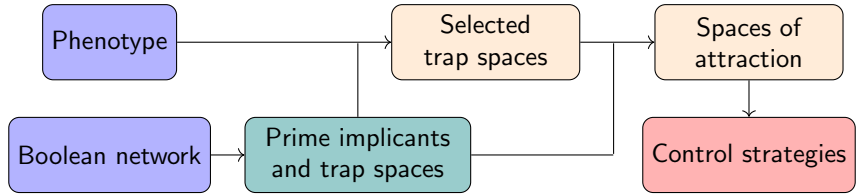
\begin{figure}[h]
\centering
\begin{tikzpicture}[node distance=1.5cm, every node/.style={fill=white, font=\sffamily}, align=center]
\node (Pheno) [rectangle, rounded corners, draw=black, minimum height=1cm, text centered, fill=blue!30] {Phenotype};
\node (BN) [rectangle, rounded corners, draw=black, minimum height=1cm, text centered, fill=blue!30, below of = Pheno] {Boolean network};
\node (Blank) [rectangle, draw=none, right of= Pheno, yshift=0.11cm, xshift=1.8cm] {};
\node (PIGTS) [rectangle, rounded corners, draw=black, minimum width= 3cm, minimum height=1cm, text centered, fill=teal!40, below of = Blank, yshift=-0.11cm] {Prime implicants \\ and trap spaces};
\node (selTS)  [rectangle, rounded corners, draw=black, minimum height=1cm, text centered, minimum width=2.5cm, fill=orange!15, right of= Pheno, xshift=3.5cm] {Selected \\ trap spaces};
\node (Blank2) [rectangle, draw=none, right of= selTS, yshift=0.11cm, xshift=0.2cm] {};
\node (Blank3) [rectangle, draw=none, right of= PIGTS, xshift=2.02cm] {};
\node (Blank4) [rectangle, draw=none, below of= Blank2, yshift=-0.22cm] {};
\node (SoA) [rectangle, rounded corners, draw=black, minimum height=1cm, text centered, minimum width=2.5cm, fill=orange!15, right of=selTS, xshift=2cm]   {Spaces of \\ attraction};
\node (CS) [rectangle, rounded corners, draw=black, minimum height=1cm, text centered, fill=red!30, below of=SoA] {Control strategies};
\draw[->] (BN) -- (PIGTS);    
\draw[-] (PIGTS) -- (Blank);
\draw[->] (Pheno) -- (selTS);
\draw[-] (PIGTS) -- (Blank3);
\draw[-] (Blank4) -- (Blank2);
\draw[->] (selTS) -- (SoA);
\draw[->] (SoA) -- (CS);
\end{tikzpicture}
\caption{Main steps of the method for finding control strategies for a phenotype, represented in color boxes according to their role: inputs (blue), precomputation (green), main computation (beige), output (red).}\label{fig:method}
\end{figure}

We describe now the main steps of the method, outlined in \Cref{fig:method}.

\begin{algorithm}
\renewcommand{\thealgorithm}{1}
\caption{Control strategies for a phenotype P}\label{alg:ca}
\hspace*{\algorithmicindent} \textbf{Input}: $f$ Boolean function, $P$ phenotype, $attr$ attractors of $f$ (optional) \\
\hspace*{\algorithmicindent} $m$ limit size of the control strategies (optional) \\
\hspace*{\algorithmicindent} \textbf{Output}: control strategies for $P$
\begin{algorithmic}[1]
\Function{ControlStrategies}{$f$, $P$, $attr$}
	\State \textbf{T} $\gets$ trapSpaces($f$)
	\State \textbf{selTS} $\gets$ selectedTrapSpaces1(\textbf{T}, $P$)
	\If {$attr \neq \emptyset$}:
		\State \textbf{selTS} $\gets$ selTS + selectedTrapSpaces2(\textbf{T}, $P$, $attr$)
	\EndIf
	\State \textbf{CA} $\gets \emptyset$
	\For{$i$ in $\{1, \dots$, min($m$, $n$)$\}$}: \Comment{$n$ total number of variables}
		\State \textbf{S} $\gets$ $\{$S subspace$\colon|$fixed(S)$|$ = $i$, $\exists$T $\in$ \textbf{selTS} with T $\subseteq$ S$\}$
		\For {S in \textbf{S}}:
			\If {(S $\not\subseteq$ S' for all S' in \textbf{CA}) \textbf{and} isSpaceAttraction($f$, S, \textbf{selTS})}:
				\State add S to \textbf{CA}
			\EndIf
		\EndFor
	\EndFor
	\State \Return \textbf{CA}
\EndFunction
\end{algorithmic}
\end{algorithm}		

\begin{algorithm}
\renewcommand{\thealgorithm}{2}
\caption{Subspace is a space of attraction}\label{alg:isSoa}
\hspace*{\algorithmicindent} \textbf{Input}: $f$ Boolean function, S subspace, \textbf{TS} trap spaces \\
\hspace*{\algorithmicindent} \textbf{Output}: \textit{True} if S is space of attraction of a trap space in \textbf{TS}. \textit{False} otherwise.
\begin{algorithmic}[1]
\Function{IsSpaceAttraction}{f, S, \textbf{TS}}
	\State f' $\gets$ percolateFunction(f, S)
	\State \Return isNotEmpty($\{$T in \textbf{TS}: T $\subseteq$ S \textbf{and} fixed(T) $\subseteq$ fixed(f')$\}$)
\EndFunction
\end{algorithmic}
\end{algorithm}	

\textbf{Inputs.} The inputs are the Boolean function describing the system and the subspace of the target phenotype $P$. The attractors, if known, are also used as input. Prime implicants and trap spaces can be given as input or computed from the Boolean function. 

\textbf{Selection of trap spaces.} Trap spaces of interest are divided into two types: trap spaces contained in $P$ (Type 1) and trap spaces not contained in $P$ but containing only attractors in $P$ (Type 2). As trap spaces have been identified in the previous step, this selection only requires checking whether a trap space belongs to one of the types (\Cref{alg:ca}: 3-5). Trap spaces of Type 2 are only identified when all the attractors are known or can be approximated by minimal trap spaces. In order to avoid unnecessary calculations, we do not consider trap spaces that percolate to smaller ones, since if a trap space $T_1$ percolates to a trap space $T_2$, all spaces of attractions of $T_1$ are also spaces of attraction of $T_2$.

\textbf{Computation of spaces of attraction.} 
Spaces of attraction for the trap spaces from the previous step are computed using the theoretical principles described in \Cref{perc_soa}. The detailed procedure is shown in \Cref{alg:ca}: 6-11.
For each subspace $S$ that contains at least one of the selected trap spaces (\Cref{alg:ca}: 8), it is checked whether it is a space of attraction for one of the selected trap spaces (\Cref{alg:ca}: 10). To do so, the percolated function of $f$ obtained by fixing the variables in $S$ is calculated (\Cref{alg:isSoa}: 2). If $T$ is contained in the subspace generated by $S$ and all the variables fixed in $T$ are also fixed in the percolated function, then the subspace generated by $S$ is a space of attraction of $T$ (\Cref{alg:isSoa}: 3). Since the aim is to find maximal spaces of attraction satisfying this property, the subspaces $S$ are taken randomly fixing an increasing number of variables, so that supersets of sets already defining a space of attraction are not considered (\Cref{alg:ca}: 8, 10).

\textbf{Output} The obtained spaces of attraction are control strategies for the phenotype $P$ by \Cref{prop_CS} and, therefore, are returned as output.
		
The method also allows to include some constraints on the control strategies. One example is the exclusion of some components, which can be taken into account when selecting the subspaces $S$ (\Cref{alg:ca}: 8). Another constraint to consider is on the size of the control strategies. Imposing a limit on the number of interventions might allow to reduce the computational cost without losing interesting solutions, since small control strategies are usually the most relevant.

\section{Application: cell fate decision networks} \label{Application}

In this section we discuss the application of our method to two Boolean networks describing cell fate decision processes. In the first case study we consider two different control problems, one having a phenotype as target for the control, the second targeting single attractors. The second case study focuses on phenotype control. We compare the control strategies identified by our approach to the ones obtained using exclusively value percolation, as described in \Cref{sub:controlstrategies}. We show that, for both examples, new control strategies can be identified with the procedure introduced in this work. 

All computations in this section were done on an 8-processor computer, Intel(R)Core(TM) i7-2600 CPU at 3.40GHz, 16GB memory, without any use of parallelization.

\subsection{MAPK network}

The network considered in this case study was introduced by Grieco et al. (2013) \cite{MAPK_network} to model the effect of the Mitogen-Activated Protein Kinase (MAPK) pathway on cell fate decisions taken in pathological cells (see \Cref{fig:mapk}). It uses 53 Boolean variables, four being inputs (DNA-damage, EGFR-stimulus, FGFR3-stimulus and TGFBR-stimulus) and three outputs (Apoptosis, Proliferation and Growth-Arrest).

The asynchronous dynamics has 18 attractors, 12 being stable states and 6 cyclic attractors. All of them can be approximated by minimal trap spaces, since each minimal trap space only contains one attractor and there is no attractor that is not contained in a minimal trap space \cite{AttractorApprox_Klarner}. Therefore, we can use trap spaces of both Type 1 and Type 2 to compute control strategies.

The phenotype chosen as target for the control is the apoptosis phenotype, which is defined in \cite{MAPK_network} as the states fixing Apoptosis and Growth Arrest to 1 and Proliferation to 0. There are 103 non-percolating trap spaces, which are trap spaces that do not percolate to smaller ones, containing only attractors in the apoptosis phenotype. Of these, 64 are of Type 1 and 39 of Type 2. We set an upper bound of four components to the size of the control strategies, since generally only small control strategies are of interest and this limit already allows to find relevant ones. In addition, we exclude interventions that fix any of the output nodes of the network. In this setting, we identify two control strategies of size 1 ($\{$TGFBR-stimulus = 1$\}$ and $\{$DNA-damage = 1$\}$) and no control strategies of size 2, 3 and 4. The running time is around 13 minutes.

Using exclusively the percolation of the fixed values we identify two control strategies of size 1 ($\{$TGFBR-stimulus = 1$\}$ and $\{$TGFBR = 1$\}$), 121 control strategies of size 2, 164 of size 3 and 139 of size 4. Looking at the Boolean function, we observe that TGFBR is uniquely regulated by TGFBR-stimulus, so fixing TGFBR-stimulus to 1 implies that TGFBR is also fixed to 1 and, therefore, these interventions are equivalent in terms of their effect on the apoptosis phenotype. However, it is obvious that if the control fixing TGFBR to 1 is released, TGFBR could be updated to zero again by TGFBR-stimulus, and this change would induce the system to leave the apoptosis phenotype. Therefore, the control of TGFBR requires a permanent intervention.

Our method uncovers the control strategy $\{$DNA-damage = 1$\}$, which is not obtained by using solely value percolation. In fact, the percolation of the subspace defined by this strategy does not reach the phenotype, but stops at the subspace $T = \{$DNA-damage = 1, ATM = 1, TAOK = 1$\}$. However, since $T$ is one of the trap spaces selected by our method, the constraint $\{$DNA-damage = 1$\}$ is identified as a control strategy.

Of the control strategies of size 2, 3 and 4 that can be identified by percolation, 18, 13 and 7 respectively are supersets of the control strategy $\{$DNA-damage = 1$\}$ identified by our method. For this reason, the subspaces obtained by percolating these interventions are contained in the trap space $T$ mentioned above and therefore the associated control can be eventually released, without affecting the reachability of the target. The remaining control strategies are not guaranteed to lead to a trap space. As a consequence, in these cases, an early release of the control could lead to the loss of the control goal. This illustrates how our method can complement previous approaches, by identifying control strategies of reduced complexity, and, consequently, reducing the number of interventions to be considered, while at the same time providing information about the effects of a possible release of the control.

The components appearing in the minimal control strategies identified (DNA-damage and TGFBR-stimulus) correspond to two inputs of the model. These inputs represent anti-proliferative stumuli from the MAPK network \cite{MAPK_network} and, therefore, can be expected to play an important role in the phenotype decision. It is, however, certainly interesting that they are capable of fully inducing the apoptosis phenotype without further conditions on internal processes.

In addition to the control problem for the apoptosis phenotype, we also searched for control strategies for the 10 apoptotic stable states. We set the maximum size of control strategies to five. For eight stable states ($A_1$ to $A_8$ in \Cref{tab:stablestates}) exactly one control strategy of size 4 is obtained. For stable state $A_9$, two control strategies of size 5 are found, and for $A_{10}$ no control strategies up to size 5 are identified. The list of stable states and their control strategies can be found in the supplementary material. The running time for one stable state is around 26 minutes. 

Since the chosen stable states belong to the apoptosis phenotype, all the selected trap spaces are also considered when computing the control strategies for the apoptosis phenotype. Therefore, the control strategies of the stable states are subspaces of the ones obtained for the apoptosis phenotype. One of the main differences is that the four inputs are present in all the control strategies of the stable states. The input variables are, by definition, not regulated by any component, and therefore must be directly controlled if the value in a given steady state is to be achieved. The analysis of the control problem for the phenotype revealed that fixing DNA-damage to 1 is enough to lead the system to the apoptosis subspace, but fixing the additional inputs is necessary to obtain a specific steady state. 
Fixing the four inputs is already enough to induce the stable states $A_1$ to $A_8$ solely by percolation. However, the stable states $A_9$ and $A_{10}$ require additional internal processes to be controlled. For $A_9$, the two control strategies identified do not percolate directly to the attractor, but lead the dynamics to one of the selected trap spaces. For $A_{10}$, no control strategies up to size 5 are found neither by our method nor percolation techniques, suggesting that a higher number of interventions might be necessary. These observations show that control for a phenotype can be more achievable than for a specific attractor, and thus in some cases more interesting for application.

\subsection{T-LGL network}

We now consider a control problem for the network introduced by Zhang  et al. (2008) \cite{TLGL_network} to model the T cell large granular lymphocite (T-LGL) survival signaling network (see \Cref{fig:tlgl}). It consists of 60 Boolean variables, six being inputs (CD45, IL15, PDGF, Stimuli, Stimuli2 and TAX) and three readouts (Apoptosis, Proliferation and Cytoskeleton-signaling).

The asynchronous dynamics has 156 attractors, 86 being stable states and 70 cyclic attractors. As in the previous network, all of them can be approximated by minimal trap spaces \cite{AttractorApprox_Klarner}. Therefore, we can use trap spaces of both Type 1 and Type 2 to compute control strategies. 

We consider the apoptosis phenotype defined by fixing Apoptosis to 1 and Proliferation to 0. Note that the third readout, Cytoskeleton signaling, is forced to 0 by its regulator Apoptosis having value 1. There are 883 non-percolating trap spaces, which are trap spaces that do not percolate to smaller ones, containing only attractors in the apoptosis phenotype. 729 trap spaces are of Type 1 and 154 of Type 2. As in the previous case study, we set an upper bound of four components to the size of the control strategies and we exclude interventions that fix any of the readout nodes of the network. In this setting, six control strategies are identified: three of size 3 ($\{$CD45 = 0, IL15 = 0, PDGF = 1$\}$, $\{$CD45 = 0, IL15 = 0, Stimuli = 1$\}$, $\{$CD45 = 0, IL15 = 0, TAX = 1$\}$) and three of size 4 ($\{$CD45 = 1, PDGF = 0, PDGFR = 0, Stimuli2 = 1$\}$, $\{$CD45 = 1, PDGF = 0, S1P = 0, Stimuli2 = 1$\}$, $\{$CD45 = 1, PDGF = 0, SPHK1 = 0, Stimuli2 = 1$\}$). The running time is around 15 minutes.

The three control strategies of size 3 consist only of input components. All the control strategies of size 4 have three components in common while the fourth varies within PDGFR, S1P and SPHK1, suggesting that these three interventions might be equivalent in terms of their effect on the apoptosis phenotype. In fact, by looking at the Boolean function, we observe that fixing PDGFR = 0, implies SPHK1 = 0, which also implies S1P = 0. Identifying such equivalent interventions a priori might allow to reduce the computational cost of the method. 

Using only percolation we find exactly one control strategy of size 1 ($\{$Caspase = 1$\}$) and none of size 2, 3 or 4. However, this control strategy is relatively trivial since the Caspase component is directly regulating Apoptosis. The control strategies identified by our method do not percolate directly to the phenotype. At the end of the percolation process, the dynamics reaches one of the trap spaces selected as containing only attractors in the apoptosis phenotype. \\

This case study highlights the added value of our approach which can uncover relevant system interventions that are not identified by usual percolation approaches.

\section{Discussion}

In this work, we considered properties of trap spaces and principles of target control to introduce a new approach to compute control strategies. The procedure proposed is applicable to both phenotype and attractor control and allows the interventions to be released after a certain amount of time, in contrast to usual target control methods that require permanent interventions.

The approach presented here is widely applicable to Boolean models of biological systems and can provide intervention strategies that are independent of the type of update considered in the modeling. Moreover, restrictions on the control strategies, in the form of variables to be excluded, can be added. Our approach also allows to incorporate information about the attractors, with the possibility to obtain control strategies that escape regular percolation-based techniques. As demonstrated with the two case studies, our method can identify new control strategies that require a small number of control variables, and can thus reveal valuable intervention approaches.

Our approach efficiently identifies control strategies for relatively large biological networks. A naturally important further step is a rigorous comparison with existing methods, for instance approaches based on stable motifs \cite{TargetControl,ControlMotifs}.
Furthermore, the performance of the method could benefit from the adoption of fine-tuning strategies developed to speed up some of the procedures involved in candidate screening. For instance, we could consider the reduction of the size of the search space by identifying a priori equivalent interventions, adapting existing approaches \cite{InterventionSets}. Further steps also include the extension of the method to other types of control, such as edge interventions or sequential control.

\section*{Acknowledgments}
E.T. was funded by the Volkswagen Stiftung (Volkswagen Foundation) under the funding initiative Life? - A fresh scientific approach to the basic principles of life (project ID: 93063).

\bibliographystyle{splncs04}
\bibliography{references}

\begin{thebibliography}{10}
\providecommand{\url}[1]{\texttt{#1}}
\providecommand{\urlprefix}{URL }
\providecommand{\doi}[1]{https://doi.org/#1}

\bibitem{ApoptosisCancer}
Baig, S., Seevasant, I., Mohamad, J., Mukheem, A., Huri, H.Z., Kamarul, T.:
  Potential of apoptotic pathway-targeted cancer therapeutic research: Where do
  we stand? Cell Death \& Disease  \textbf{7}(1),  e2850 (2016).
  \doi{10.1038/cddis.2015.275}

\bibitem{ControlBCN}
{Biane}, C., {Delaplace}, F.: Causal reasoning on boolean control networks
  based on abduction: Theory and application to cancer drug discovery. IEEE/ACM
  Transactions on Computational Biology and Bioinformatics  \textbf{16}(5),
  1574--1585 (2019). \doi{10.1109/TCBB.2018.2889102}

\bibitem{CellFate_network}
Calzone, L., Tournier, L., Fourquet, S., Thieffry, D., Zhivotovsky, B.,
  Barillot, E., Zinovyev, A.: Mathematical modelling of cell-fate decision in
  response to death receptor engagement. PLOS Computational Biology
  \textbf{6}(3),  1--15 (2010). \doi{10.1371/journal.pcbi.1000702}

\bibitem{GINsim}
Chaouiya, C., Naldi, A., Thieffry, D.: Logical Modelling of Gene Regulatory
  Networks with GINsim., vol.~804, pp. 463--79 (2012)

\bibitem{DrugDiscovery}
Csermely, P., Korcsmáros, T., Kiss, H.J., London, G., Nussinov, R.: Structure
  and dynamics of molecular networks: A novel paradigm of drug discovery: A
  comprehensive review. Pharmacology \& Therapeutics  \textbf{138}(3),  333 --
  408 (2013). \doi{10.1016/j.pharmthera.2013.01.016}

\bibitem{Sinergies}
Flobak, {\AA}., Baudot, A., Remy, E., Thommesen, L., Thieffry, D., Kuiper, M.,
  L{\ae}greid, A.: Discovery of drug synergies in gastric cancer cells
  predicted by logical modeling. PLOS Computational Biology  \textbf{11}(8),
  1--20 (2015). \doi{10.1371/journal.pcbi.1004426}

\bibitem{MAPK_network}
Grieco, L., Calzone, L., Bernard-Pierrot, I., Radvanyi, F., Kahn-Perlès, B.,
  Thieffry, D.: Integrative modelling of the influence of mapk network on
  cancer cell fate decision. PLOS Computational Biology  \textbf{9}(10),  1--15
  (10 2013). \doi{10.1371/journal.pcbi.1003286}

\bibitem{KernelControl}
Kim, J., Park, S.M., Cho, K.H.: Discovery of a kernel for controlling
  biomolecular regulatory networks. Scientific Reports  \textbf{3}, ~2223
  (2013). \doi{10.1038/srep02223}

\bibitem{Implicants_Klarner}
Klarner, H., Bockmayr, A., Siebert, H.: Computing maximal and minimal trap
  spaces of boolean networks. Natural Computing  \textbf{14},  535--544 (2015).
  \doi{10.1007/s11047-015-9520-7}

\bibitem{AttractorApprox_Klarner}
Klarner, H., Siebert, H.: Approximating attractors of boolean networks by
  iterative ctl model checking. Frontiers in Bioengineering and Biotechnology
  \textbf{3}, ~130 (2015). \doi{10.3389/fbioe.2015.00130}

\bibitem{PyBoolNet}
Klarner, H., Streck, A., Siebert, H.: {PyBoolNet: a python package for the
  generation, analysis and visualization of boolean networks}. Bioinformatics
  \textbf{33}(5),  770--772 (2016). \doi{10.1093/bioinformatics/btw682}

\bibitem{Controllability_Networks}
Liu, Y.Y., Slotine, J.J., Barabási, A.L.: Controllability of complex networks.
  Nature  \textbf{473},  167--173 (2011). \doi{10.1038/nature10011}

\bibitem{ControlBasins}
Mandon, H., Su, C., Haar, S., Pang, J., Paulev{\'e}, L.: Sequential
  reprogramming of boolean networks made practical. In: Bortolussi, L.,
  Sanguinetti, G. (eds.) Computational Methods in Systems Biology. vol. 11773,
  pp. 3--19. Springer International Publishing, Cham (2019).
  \doi{10.1007/978-3-030-31304-3\_1}

\bibitem{ControlAlgebra}
Murrugarra, D., Veliz-Cuba, A., Aguilar, B., Laubenbacher, R.: Identification
  of control targets in boolean molecular network models via computational
  algebra. BMC Systems Biology  \textbf{10}(1), ~94 (2016).
  \doi{10.1186/s12918-016-0332-x}

\bibitem{InterventionSets}
Samaga, R., Kamp, A.V., Klamt, S.: Computing combinatorial intervention
  strategies and failure modes in signaling networks. Journal of Computational
  Biology  \textbf{17}(1),  39--53 (2010). \doi{10.1089/cmb.2009.0121}

\bibitem{Reprogramming}
Takahashi, K., Yamanaka, S.: A decade of transcription factor-mediated
  reprogramming to pluripotency. Nature Reviews Molecular Cell Biology
  \textbf{17}(3),  183--193 (2016). \doi{10.1038/nrm.2016.8}

\bibitem{TargetControl}
Yang, G., Gómez Tejeda~Zañudo, J., Albert, R.: Target control in logical
  models using the domain of influence of nodes. Frontiers in Physiology
  \textbf{9}, ~454 (2018). \doi{10.3389/fphys.2018.00454}

\bibitem{FeedbackVertexSet}
Za{\~n}udo, J.G.T., Yang, G., Albert, R.: Structure-based control of complex
  networks with nonlinear dynamics. Proceedings of the National Academy of
  Sciences  \textbf{114}(28),  7234--7239 (2017). \doi{10.1073/pnas.1617387114}

\bibitem{ControlMotifs}
Zañudo, J.G.T., Albert, R.: Cell fate reprogramming by control of
  intracellular network dynamics. PLOS Computational Biology  \textbf{11}(4),
  1--24 (2015). \doi{10.1371/journal.pcbi.1004193}

\bibitem{TLGL_network}
Zhang, R., Shah, M.V., Yang, J., Nyland, S.B., Liu, X., Yun, J.K., Albert, R.,
  Loughran, T.P.: Network model of survival signaling in large granular
  lymphocyte leukemia. Proceedings of the National Academy of Sciences
  \textbf{105}(42),  16308--16313 (2008). \doi{10.1073/pnas.0806447105}

\end{thebibliography}

\section{Supplementary material}\label{supplementary}

\begin{figure}[h]
\centering
\includegraphics[width=\textwidth]{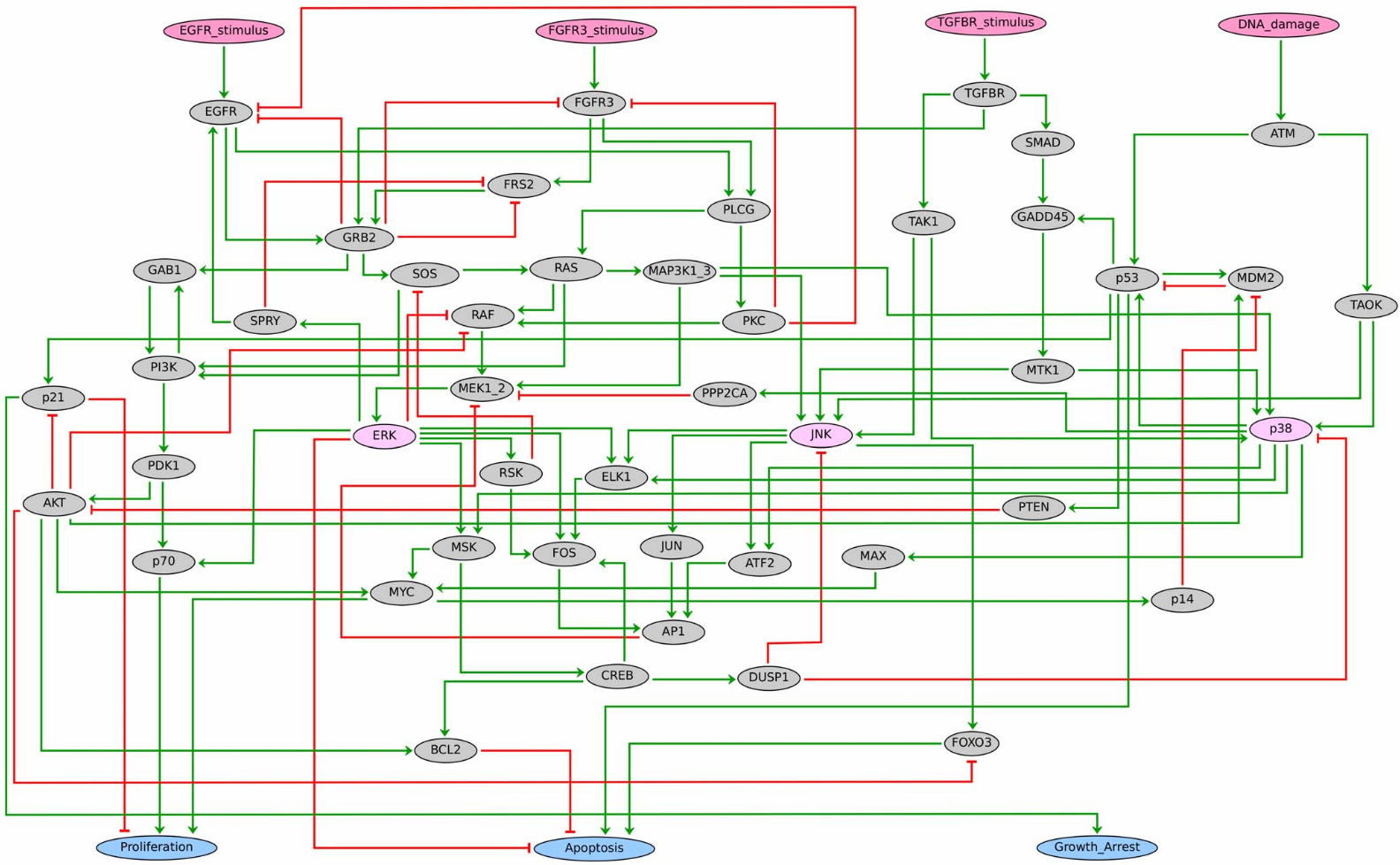}
\caption{MAPK network, figure adapted from \cite{MAPK_network}.}
\label{fig:mapk}
\end{figure}

\begin{table}
\caption{Apoptotic stable states of MAPK network (10). Each column represents a stable state. The number in the cell indicates the value of the variable in the stable state.}
\label{tab:stablestates}
\centering
\scriptsize
\begin{tabular}{|m{2.3cm}|>{\centering\arraybackslash}m{0.5cm}|>{\centering\arraybackslash}m{0.5cm}|>{\centering\arraybackslash}m{0.5cm}|>{\centering\arraybackslash}m{0.5cm}|>{\centering\arraybackslash}m{0.5cm}|>{\centering\arraybackslash}m{0.5cm}|>{\centering\arraybackslash}m{0.5cm}|>{\centering\arraybackslash}m{0.5cm}|>{\centering\arraybackslash}m{0.5cm}|>{\centering\arraybackslash}m{0.5cm}|}
\hline
 & $A_1$ & $A_2$ & $A_3$ & $A_4$ & $A_5$ & $A_6$ & $A_7$ & $A_8$ & $A_9$ & $A_{10}$ \\
\hline 
AKT & 0 & 0 & 0 & 0 & 0 & 0 & 0 & 0 & 0 & 0 \\
\hline 
AP1 & 1 & 1 & 1 & 1 & 1 & 1 & 1 & 1 & 1 & 1 \\
\hline 
ATF2 & 1 & 1 & 1 & 1 & 1 & 1 & 1 & 1 & 1 & 1 \\
\hline 
ATM & 0 & 1 & 0 & 1 & 0 & 0 & 1 & 1 & 1 & 1 \\
\hline 
Apoptosis & 1 & 1 & 1 & 1 & 1 & 1 & 1 & 1 & 1 & 1 \\
\hline 
BCL2 & 0 & 0 & 0 & 0 & 0 & 0 & 0 & 0 & 0 & 0 \\
\hline 
CREB & 1 & 1 & 1 & 1 & 1 & 1 & 1 & 1 & 1 & 1 \\
\hline 
DNA-damage & 0 & 1 & 0 & 1 & 0 & 0 & 1 & 1 & 1 & 1 \\
\hline 
DUSP1 & 1 & 1 & 1 & 1 & 1 & 1 & 1 & 1 & 1 & 1 \\
\hline 
EGFR & 0 & 0 & 0 & 0 & 0 & 0 & 0 & 0 & 0 & 0 \\
\hline 
EGFR-stimulus & 1 & 1 & 0 & 0 & 1 & 0 & 1 & 0 & 0 & 0 \\
\hline 
ELK1 & 1 & 1 & 1 & 1 & 1 & 1 & 1 & 1 & 1 & 1 \\
\hline 
ERK & 0 & 0 & 0 & 0 & 0 & 0 & 0 & 0 & 0 & 0 \\
\hline 
FGFR3 & 0 & 0 & 0 & 0 & 0 & 0 & 0 & 0 & 0 & 0 \\
\hline 
FGFR3-stimulus & 1 & 1 & 1 & 1 & 0 & 0 & 0 & 0 & 0 & 0 \\
\hline 
FOS & 0 & 0 & 0 & 0 & 0 & 0 & 0 & 0 & 0 & 0 \\
\hline 
FOXO3 & 1 & 1 & 1 & 1 & 1 & 1 & 1 & 1 & 1 & 1 \\
\hline 
FRS2 & 0 & 0 & 0 & 0 & 0 & 0 & 0 & 0 & 0 & 0 \\
\hline 
GAB1 & 1 & 1 & 1 & 1 & 1 & 1 & 1 & 1 & 1 & 0 \\
\hline 
GADD45 & 1 & 1 & 1 & 1 & 1 & 1 & 1 & 1 & 1 & 1 \\
\hline 
GRB2 & 1 & 1 & 1 & 1 & 1 & 1 & 1 & 1 & 0 & 0 \\
\hline 
Growth-Arrest & 1 & 1 & 1 & 1 & 1 & 1 & 1 & 1 & 1 & 1 \\
\hline 
JNK & 1 & 1 & 1 & 1 & 1 & 1 & 1 & 1 & 1 & 1 \\
\hline 
JUN & 1 & 1 & 1 & 1 & 1 & 1 & 1 & 1 & 1 & 1 \\
\hline 
MAP3K1-3 & 1 & 1 & 1 & 1 & 1 & 1 & 1 & 1 & 0 & 0 \\
\hline 
MAX1 & 1 & 1 & 1 & 1 & 1 & 1 & 1 & 1 & 1 & 1 \\
\hline 
MDM2 & 0 & 0 & 0 & 0 & 0 & 0 & 0 & 0 & 0 & 0 \\
\hline 
MEK1-2 & 0 & 0 & 0 & 0 & 0 & 0 & 0 & 0 & 0 & 0 \\
\hline 
MSK & 1 & 1 & 1 & 1 & 1 & 1 & 1 & 1 & 1 & 1 \\
\hline 
MTK1 & 1 & 1 & 1 & 1 & 1 & 1 & 1 & 1 & 1 & 1 \\
\hline 
MYC & 1 & 1 & 1 & 1 & 1 & 1 & 1 & 1 & 1 & 1 \\
\hline 
PDK1 & 1 & 1 & 1 & 1 & 1 & 1 & 1 & 1 & 1 & 0 \\
\hline 
PI3K & 1 & 1 & 1 & 1 & 1 & 1 & 1 & 1 & 1 & 0 \\
\hline 
PKC & 0 & 0 & 0 & 0 & 0 & 0 & 0 & 0 & 0 & 0 \\
\hline 
PLCG & 0 & 0 & 0 & 0 & 0 & 0 & 0 & 0 & 0 & 0 \\
\hline 
PPP2CA & 1 & 1 & 1 & 1 & 1 & 1 & 1 & 1 & 1 & 1 \\
\hline 
PTEN & 1 & 1 & 1 & 1 & 1 & 1 & 1 & 1 & 1 & 1 \\
\hline 
Proliferation & 0 & 0 & 0 & 0 & 0 & 0 & 0 & 0 & 0 & 0 \\
\hline 
RAF & 1 & 1 & 1 & 1 & 1 & 1 & 1 & 1 & 0 & 0 \\
\hline 
RAS & 1 & 1 & 1 & 1 & 1 & 1 & 1 & 1 & 0 & 0 \\
\hline 
RSK & 0 & 0 & 0 & 0 & 0 & 0 & 0 & 0 & 0 & 0 \\
\hline 
SMAD & 1 & 1 & 1 & 1 & 1 & 1 & 1 & 1 & 0 & 0 \\
\hline 
SOS & 1 & 1 & 1 & 1 & 1 & 1 & 1 & 1 & 0 & 0 \\
\hline 
SPRY & 0 & 0 & 0 & 0 & 0 & 0 & 0 & 0 & 0 & 0 \\
\hline 
TAK1 & 1 & 1 & 1 & 1 & 1 & 1 & 1 & 1 & 0 & 0 \\
\hline 
TAOK & 0 & 1 & 0 & 1 & 0 & 0 & 1 & 1 & 1 & 1 \\
\hline 
TGFBR & 1 & 1 & 1 & 1 & 1 & 1 & 1 & 1 & 0 & 0 \\
\hline 
TGFBR-stimulus & 1 & 1 & 1 & 1 & 1 & 1 & 1 & 1 & 0 & 0 \\
\hline 
p14 & 1 & 1 & 1 & 1 & 1 & 1 & 1 & 1 & 1 & 1 \\
\hline 
p21 & 1 & 1 & 1 & 1 & 1 & 1 & 1 & 1 & 1 & 1 \\
\hline 
p38 & 1 & 1 & 1 & 1 & 1 & 1 & 1 & 1 & 1 & 1 \\
\hline 
p53 & 1 & 1 & 1 & 1 & 1 & 1 & 1 & 1 & 1 & 1 \\
\hline 
p70 & 0 & 0 & 0 & 0 & 0 & 0 & 0 & 0 & 0 & 0 \\
\hline
\end{tabular}
\end{table}

\begin{table}
\caption{Control strategies up to size 5 obtained for the apoptotic stable states of the MAPK network. Each column represents a control strategy for the indicated stable state. A number in a cell indicates the value to which the variable is fixed in the control strategy. Empty cells denote uncontrolled components.}
\label{tab:cs}
\centering
\begin{tabular}{|m{2.9cm}|>{\centering\arraybackslash}m{0.5cm}|>{\centering\arraybackslash}m{0.5cm}|>{\centering\arraybackslash}m{0.5cm}|>{\centering\arraybackslash}m{0.5cm}|>{\centering\arraybackslash}m{0.5cm}|>{\centering\arraybackslash}m{0.5cm}|>{\centering\arraybackslash}m{0.5cm}|>{\centering\arraybackslash}m{0.5cm}|>{\centering\arraybackslash}m{0.5cm}|>{\centering\arraybackslash}m{0.5cm}|}
\hline
 & $A_1$ & $A_2$ & $A_3$ & $A_4$ & $A_5$ & $A_6$ & $A_7$ & $A_8$ & \multicolumn{2}{c|}{$A_9$} \\
\hline 
DNA-damage & 0 & 1 & 0 & 1 & 0 & 0 & 1 & 1 & 1 & 1 \\
\hline 
EGFR-stimulus & 1 & 1 & 0 & 0 & 1 & 0 & 1 & 0 & 0 & 0 \\
\hline 
FGFR3-stimulus & 1 & 1 & 1 & 1 & 0 & 0 & 0 & 0 & 0 & 0 \\
\hline 
TGFBR-stimulus & 1 & 1 & 1 & 1 & 1 & 1 & 1 & 1 & 0 & 0 \\
\hline 
GAB1 &  &  &  &  &  &  &  &  & 1 &  \\
\hline 
PI3K &  &  &  &  &  &  &  &  &  & 1 \\
\hline
\end{tabular}
\end{table}

\begin{figure}
\centering
\includegraphics[width=\textwidth]{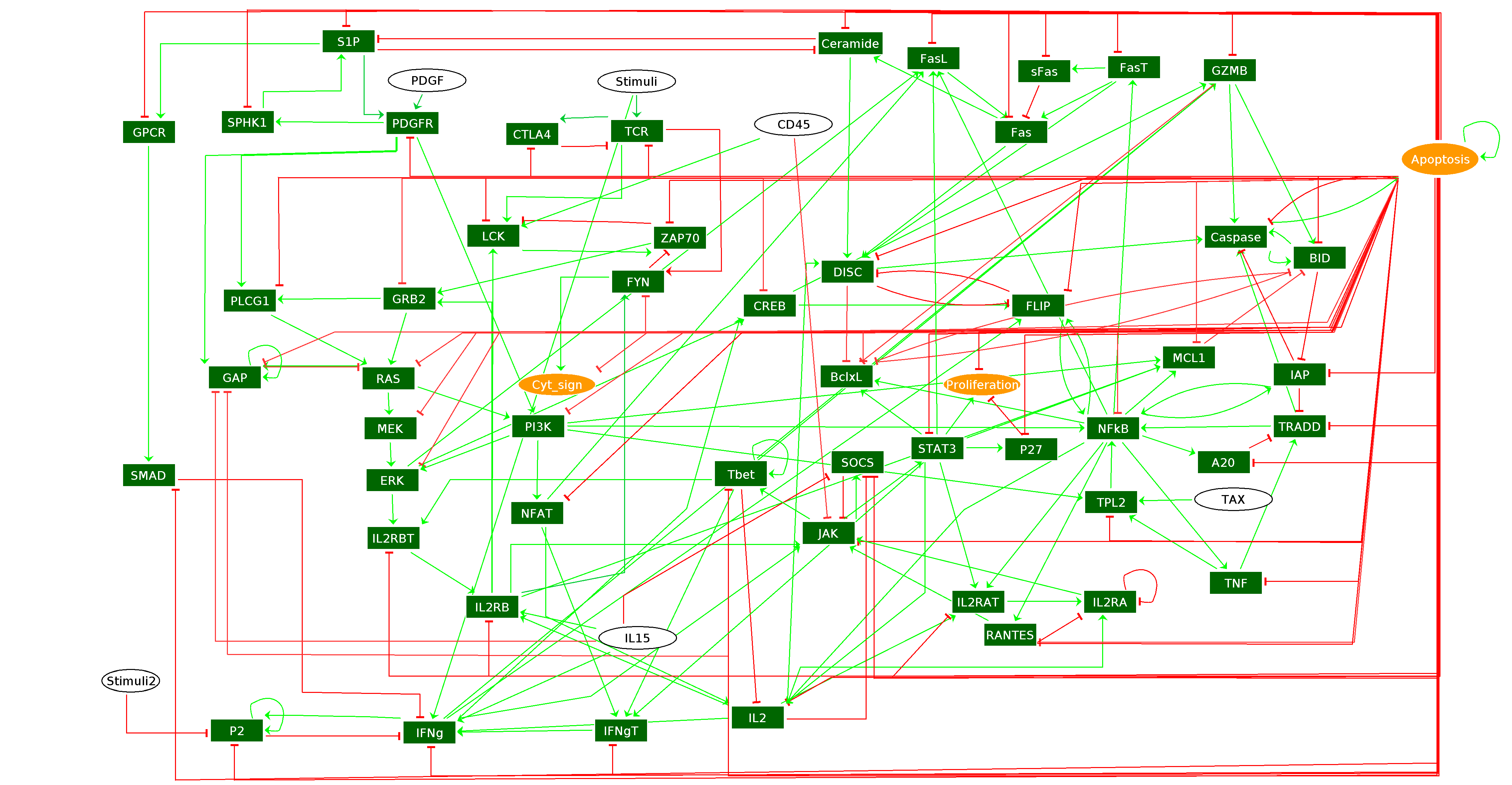}
\caption{T-LGL network, figure obtained using GINsim software \cite{GINsim}.}
\label{fig:tlgl}
\end{figure}

\end{document}